\documentclass[11pt]{article}

\usepackage{fullpage,wrapfig}
\usepackage{amsmath, amsthm, amssymb, algorithm,thmtools,algpseudocode,enumerate,caption,framed}
\usepackage{paralist, sidecap}
\usepackage[usenames,dvipsnames,svgnames,table]{xcolor}
\usepackage[colorlinks=true,pdfpagemode=UseNone,citecolor=Blue,linkcolor=Red,urlcolor=Black,pagebackref]{hyperref}
\usepackage{tikz}
\usepackage{authblk}

\newtheorem{theorem}{Theorem}[section]
\newtheorem{lemma}{Lemma}[section]

\title{Drawing HV-Restricted Planar Graphs\thanks{A preliminary version of this paper appeared in proceedings of the 11th Latin American Symposium on Theoretical Informatics (LATIN 2014)~\cite{DurocherF0M14}. Research of Stephane Durocher and Debajyoti Mondal is supported in part by the Natural Sciences and Engineering Research Council of Canada (NSERC). Research of Stefan Felsner is supported by DFG grant FE-340/11-1.}}

\author[1]{Stephane Durocher}
\author[2]{Stefan Felsner}
\author[3]{Saeed Mehrabi}
\author[4]{Debajyoti Mondal}

\affil[1]{{\small University of Manitoba, Winnipeg, Canada.
\texttt{durocher@cs.umanitoba.ca}}
}
\affil[2]{{\small Technische Universit\"at Berlin, Germany.
\texttt{felsner@math.tu-berlin.de}}
}
\affil[3]{{\small Carleton University, Ottawa, Canada.
\texttt{saeed.mehrabi@carleton.ca}}
}
\affil[4]{{\small University of Saskatchewan, Saskatoon, Canada.
\texttt{d.mondal@usask.ca}}
}

\date{}

\begin{document}

\maketitle

\begin{abstract}
A strict orthogonal drawing of a graph $G=(V, E)$ in $\mathbb{R}^2$ is a drawing of $G$ such that each vertex is mapped to a distinct point and each edge is mapped to a horizontal or vertical line segment. A graph $G$ is $HV$-restricted if each of its edges is assigned a horizontal or vertical orientation. A strict orthogonal drawing of an $HV$-restricted graph $G$ is \emph{good} if it is planar and respects the edge orientations of $G$. In this paper, we give a polynomial-time algorithm to check whether a given $HV$-restricted plane graph (i.e., a planar graph with a fixed combinatorial embedding) admits a good orthogonal drawing preserving the input embedding, which settles an open question posed by Ma\v{n}uch et al. (Graph Drawing 2010). We then examine $HV$-restricted planar graphs (i.e., when the embedding is not fixed), and give a complete characterization of the $HV$-restricted biconnected outerplanar graphs that admit good orthogonal drawings.
\end{abstract}

\section{Introduction}
\label{sec:introduction}
An \emph{orthogonal drawing} $\Gamma$ of an undirected graph $G=(V, E)$ in $\mathbb{R}^2$ is a drawing of $G$ in the plane, where each vertex of $G$ is mapped to  a distinct point and each edge of $G$ is mapped to an orthogonal polyline. $\Gamma$ is called \emph{planar} if no two edges in $\Gamma$ cross, however, two edges can meet at their common endpoints. Otherwise, the drawing is a \emph{non-planar orthogonal drawing}. An orthogonal drawing is \emph{strict} if every edge in the drawing is represented by a single vertical and horizontal line segment. Orthogonal drawings have been extensively studied over  the last two decades~\cite{DBLP:journals/comgeo/AlamKM17,battista,CornelsenK12,kant96,Tamassia87} because of its applications in many practical fields such as VLSI,  floor-planning, circuit schematics, and entity relationship diagrams.

Throughout the paper we refer to a planar graph with a fixed combinatorial embedding as a \emph{plane graph}. While drawing a plane graph, one must preserve the input embedding. But for planar graphs, any planar embedding can be chosen to draw the graph.  In 1987, Tamassia~\cite{Tamassia87} gave a polynomial-time algorithm to decide whether a plane graph  admits a strict orthogonal drawing preserving the input embedding. Later, Garg and Tamassia~\cite{garg2001} proved that deciding strict orthogonal drawability is NP-hard for planar graphs. However, polynomial-time algorithms have been developed for some well-known subclasses of planar graphs. For example, Di Battista et al.~\cite{battista} showed that the problem is polynomial-time solvable for series-parallel graphs and maximum-degree-three planar graphs. Nomura et al.~\cite{nomura} showed that a  maximum-degree-three outerplanar graph admits a planar strict orthogonal drawing if and only if it contains no cycle of three vertices.

Many variants of strict orthogonal drawings impose constraints on how the edges of the input graph have to be drawn. One of these variants describes the input graph $G$ as an \emph{$LRDU$-restricted graph} that associates each vertex-edge incidence of $G$ with an orientation (i.e., left (L), right (R), up (U), or down (D)), and asks to find an orthogonal drawing of $G$ that respects the prescribed orientations. Another variant considers \emph{$HV$-restricted graphs}, where the orientation of an edge is either horizontal (H), or vertical (V). By a \emph{good orthogonal drawing} we denote a planar strict orthogonal drawing that preserves the input edge orientations. In this paper, we only examine strict orthogonal drawings of $HV$-restricted plane and planar  graphs, and hence from now on we omit the term ``strict''.

\paragraph{$HV$-restricted plane graphs.} In 1985, Vijayan and Wigderson~\cite{vijayan1985} gave an algorithm that can decide in linear time whether an $LRDU$-restricted plane graph admits a good orthogonal  drawing,  but takes $O(n^2)$ time to construct such a drawing when it exists.  Later, Hoffmann and Kriegel~\cite{hoffmann1988} gave a linear-time construction. 

The task of characterizing $HV$-restricted plane graphs is more involved. The difficulty arises from the exponential number of choices for drawing $HV$-restricted paths, where the drawing of an $LRDU$-restricted path is unique, as illustrated in Figures~\ref{fig:intro0}(a)--(c). Recently, Ma\v{n}uch et al.~\cite{manuch2010} examined several results on the non-planar orthogonal drawings of $LRDU$- and $HV$-restricted graphs. They proved that non-planar orthogonal drawability maintaining edge orientations can be decided in polynomial-time for $LRDU$-restricted graphs, but is NP-hard for $HV$-restricted graphs. An interesting open question in this context, as posed by  Ma\v{n}uch et al.~\cite{manuch2010}, is to determine the complexity of deciding good orthogonal drawability of $HV$-restricted plane graphs. In this paper, we settle this question by giving a polynomial-time algorithm to recognize $HV$-restricted plane graphs. Here we assume that a planar embedding of the input graph is given, and our algorithm decides whether there exists a solution  that respects the input embedding.\smallskip

\paragraph{$HV$-restricted planar graphs.} A problem analogous to drawing $LRDU$-restricted graphs in $\mathbb{R}^2$ has been well studied in $\mathbb{R}^3$, but polynomial-time algorithms are known only for cycles~\cite{BattistaKLLW12} and theta graphs~\cite{GiacomoLP02}. The exponential number of possible orthogonal embeddings in  $\mathbb{R}^3$ makes the problem very difficult. Similarly, we find the problem of characterizing $HV$-restricted planar graphs that admit good orthogonal drawings  in $\mathbb{R}^2$ nontrivial even for outerplanar graphs, where the difficulty arises from the exponential number of choices for plane embeddings of the input graph.

To further illustrate the challenge, here we prove that the $HV$-restricted outerplanar graph of  Figure~\ref{fig:intro0}(d) does not admit a good orthogonal drawing. Suppose for a contradiction that $\Gamma$ is a good orthogonal drawing of $G$, and consider the drawing of the face $F=(a,b,...,f)$ in $\Gamma$. Since the edges $(a,b)$ and $(e,f)$ are horizontally oriented and $(a,f)$ is vertically oriented, either $(a,b)$ lies above $(e,f)$, or $(e,f)$ lies above $(a,b)$ in $\Gamma$. If $(a,b)$ lies above $(e,f)$ as in Figure~\ref{fig:intro0}(e), then the drawing of  cycle $a,b,i,j$  would create an edge crossing (irrespective of whether it lies inside or outside of $F$). Similarly, if $(e,f)$ lies above $(a,b)$ as in Figure~\ref{fig:intro0}(f), then the drawing of cycle $e,f,h,g$ would create an edge crossing. Drawing both of these cycles without crossing would imply a unique drawing of $F$, as shown in Figure~\ref{fig:intro0}(g). However, in this case we cannot draw the cycle $c,d,k,l$ without edge crossings.
 
\paragraph{Contributions.} We first show that the problem of whether an $HV$-restricted plane graph with $n$ vertices admits a good orthogonal drawing preserving the input embedding can be decided in $O(T(n))$ time, where $T(n)$ corresponds to the time to find a maximum flow in a multiple-source multiple-sink directed planar graph. If such a drawing exists, then it can be computed within the same running time. The best known running time for finding a maximum flow in a multiple-source multiple-sink directed planar graph is $O(n\log^2 n/\log \log n)$~\cite{MozesW10}.

We then give a linear-time characterization for biconnected outerplanar graphs that admit good orthogonal drawings. Our proof is constructive: given an $HV$-restricted outerplanar graph $G$, we can decide in linear time whether $G$ admits a good orthogonal drawing, and in $O(n^2)$-time, we can compute such a drawing if it exists. Note that the construction can choose any feasible embedding (i.e., the embedding is not fixed), and the output  is not necessarily outerplanar. A preliminary version of this paper appeared in proceedings of the 11th Latin American Symposium on Theoretical Informatics (LATIN 2014)~\cite{DurocherF0M14}. In the conference version, we only claimed to have a characterization for maximum-degree-three biconnected outerplanar graphs. 

Soon after the conference version of our paper, Didimo et al.~\cite{DidimoLP19} showed that the problem of whether an $HV$-restricted planar graph admits a good orthogonal drawing is NP-complete in general, but polynomial-time solvable for series-parallel graphs (thus also for outerplanar graphs). Their algorithm is based on a dynamic programming and runs in $O(n^4)$ time (and, in $O(n^3\log n)$ time, if the graph is of maximum-degree three). They asked whether a combinatorial characterization can be found, perhaps in terms of forbidden substructures. Our characterization for outerplanar graphs can be seen as a first step towards such a  combinatorial characterization. 

\paragraph{Organization.} In Section~\ref{sec:plane}, we show our algorithm for recognizing $HV$-restricted plane graphs. In Section~\ref{sec:planar2}, we give the necessary and sufficient conditions for a biconnected outerplanar graph to admit a good orthogonal drawing, and describe the drawing algorithm. We prove the necessity and sufficiency of those conditions in Section~\ref{sec:detail}, and  conclude the paper in Section~\ref{sec:con}.

\begin{figure}[t]
\centering
\includegraphics[width=\textwidth]{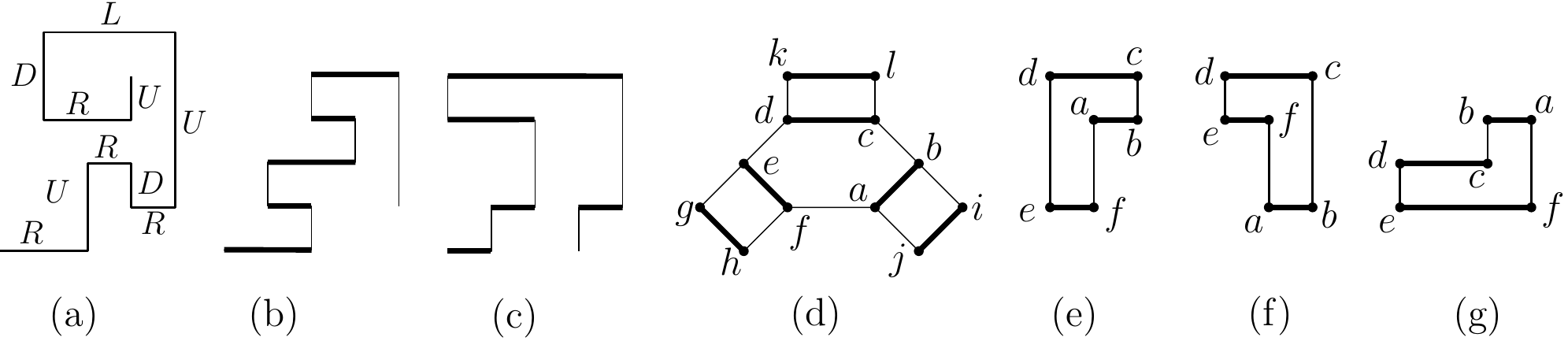}
\caption{(a) Drawing of an $LRDU$-restricted path. (b)--(c) Two different drawings of an $HV$-restricted path, where the horizontal and vertical orientations are shown in bold and thin lines, respectively. (d) An $HV$-restricted outerplanar graph $G$. (e)--(g) Drawings of the face $F$.}
\label{fig:intro0}%
\end{figure}

\section{Drawing $HV$-Restricted Plane Graphs}
\label{sec:plane}
In this section we give a polynomial-time algorithm that checks whether a given $HV$-restricted plane graph admits a good orthogonal drawing  preserving the input embedding. If the answer is affirmative, the algorithm certifies its answer by constructing a good orthogonal drawing.

We will first identify some necessary conditions and later show that they are also sufficient for the existence of the good drawing. The first condition is that every vertex has at most two incident edges with label $H$ and at most two with label $V$, and if the degree is four, the labels alternate. This condition is easily checked and from now on we assume it to be satisfied by the input.

Assume that a good drawing exists and consider a face $f$ in the
drawing.  The face is represented by a polygon, hence, if $f$ has $k$
corners, then the sum of all interior angles of  $f$ must be $(k-2)\pi$ (the
outer face makes an exception, here the angles sum to $(k+2)\pi$).
Since $f$ is an orthogonal polygon, the angle contributed by each
corner is a multiple of $\pi/2$. From the given edge orientations we
can infer the angle of some corners precisely: if a corner has two
incident edges with the same label, then it contributes an angle of
$\pi$, and if a corner corresponds to a vertex of degree one, it
contributes $2\pi$.  The interesting corners are those where the
incident edges have different labels, these corners contribute either
$\pi/2$ or $3\pi/2$. Dual to the angle condition for faces we also have the obvious
condition for vertices: around each vertex the sum of angles is
$2\pi$.

Associate a variable $x_c$ with each corner $c$ of the plane graph.
The above conditions can all be written as linear equations in these
variables. This yields a linear system $Ax=b$ and the unified
necessary condition that the system has a solution $\bar{x}$ where
each component $\bar{x}_c$ is in $\{1,2,3,4\}$. Such a solution is
 called a {\it global admissible angle assignment}.  Similar
quests for global angle assignments have been studied in rectangular
drawing problems, where Miura et al.~\cite{MiuraHN06} reduced the
problem to perfect matching, and in the context of orthogonal drawing with
bends, where Tamassia~\cite{Tamassia87} modeled an angle assignment problem
with minimum-cost maximum-flow.

Instead of directly using the  linear system stated above, we use the 
fact that the value of some variables $x_c$ is prescribed
by the input. The value for the remaining variables and hence
a global admissible angle assignment can be determined 
using a maximum-flow problem.

To construct the flow network start with the
angle graph $A(G)$ of the plane graph $G$.  The vertex set is
$V_{A(G)} = V_G \cup F_G$, i.e., the vertices of $A(G)$ are the
vertices and faces of $G$ or stated in just another way: the vertices
of $A(G)$ are the vertices of $G$ together with the vertices of the
dual $G^*$. The edges of $A(G)$ correspond to the corners of $G$: if
$v\in V_G$ and $f\in V_F$ are incident at a corner $c$ then there is an edge
$e_c=(v,f)$ in $E_{A(G)}$.  

Next step is to remove an edge $e_c=(v,f)$ from $A(G)$ when the value
of the variable $x_c$ is prescribed by the input, i.e., in the following situations:
\begin{enumerate}[(a)]
\item If the two edges of a corner have the same orientation
  and the edges are distinct, then the corner is assigned a $\pi$
  angle, i.e., $x_c=2$.
\item If the vertex corresponding to a corner is of
  degree one, then the corner is assigned a $2\pi$ angle, i.e.,
  $x_c=4$.
\item If the two edges of a corner have different orientations
  and the vertex is of degree three or more, then the corner is
  assigned a $\pi/2$ angle, i.e., $x_c=1$.  
\end{enumerate}
Let $A^\star(G)$ be the graph after removing all these edges.  Since
$A(G)$ is a plane graph the same is true for $A^\star(G)$. 
Figures~\ref{fig:flow}(a)--(b) show an example of a graph $G$ together
with the network $A^\star(G)$.

Since we want to use a fast maximum-flow algorithm, we describe the flow-problem using a planar flow network with multiple sources and sinks. It only remains to decide for some vertices of degree two in $G$ which of its corners is of size $\pi/2$ and which is of size $3\pi/2$. We model a $\pi/2$ corner with a flow of one unit entering the corresponding vertex.
 
An original vertex $v\in V_G$ is incident to an edge in $A^\star(G)$ if and only if $v$ is a vertex of degree two in $G$. With these vertices we assign a demand of $1$. The capacities of all the edges are also restricted to $1$. Finally, we have to set the excess of all $f\in F_G$. We know the total angle sum of $f$ and the angles that have been assigned in the reduction step from $A(G)$ to $A^\star(G)$. Since all the remaining angles are of size  $\pi/2$ or $3\pi/2$, we can compute how many of size $3\pi/2$ are needed, this number $z_f$ is the excess of $f$. (Note that if the computation yields a $z_f$ that is not an integer, then $G$ does not admit a good orthogonal realization). Similarly, we can also compute the number  $z'_f$ of $\pi/2$ angles that we need. For example, for the face $f_2$ in  Figure~\ref{fig:flow}(b), we consider  $3z_{f_2}+z'_{f_2}=18$ and $z_{f_2}+z'_{f_2}=10$, which solves to $(z'_{f_2},z_{f_2}) = (4,6)$. Since all edges $e_c\in E_{A^\star(G)}$ connect a source $f$ to a sink $v$, we may think of them as directed edges $f\to v$. Figure~\ref{fig:flow}(c) illustrates a maximum flow for the flow-network of Figure~\ref{fig:flow}(b).

\begin{figure}[t]
\centering
\includegraphics[width=0.7\textwidth]{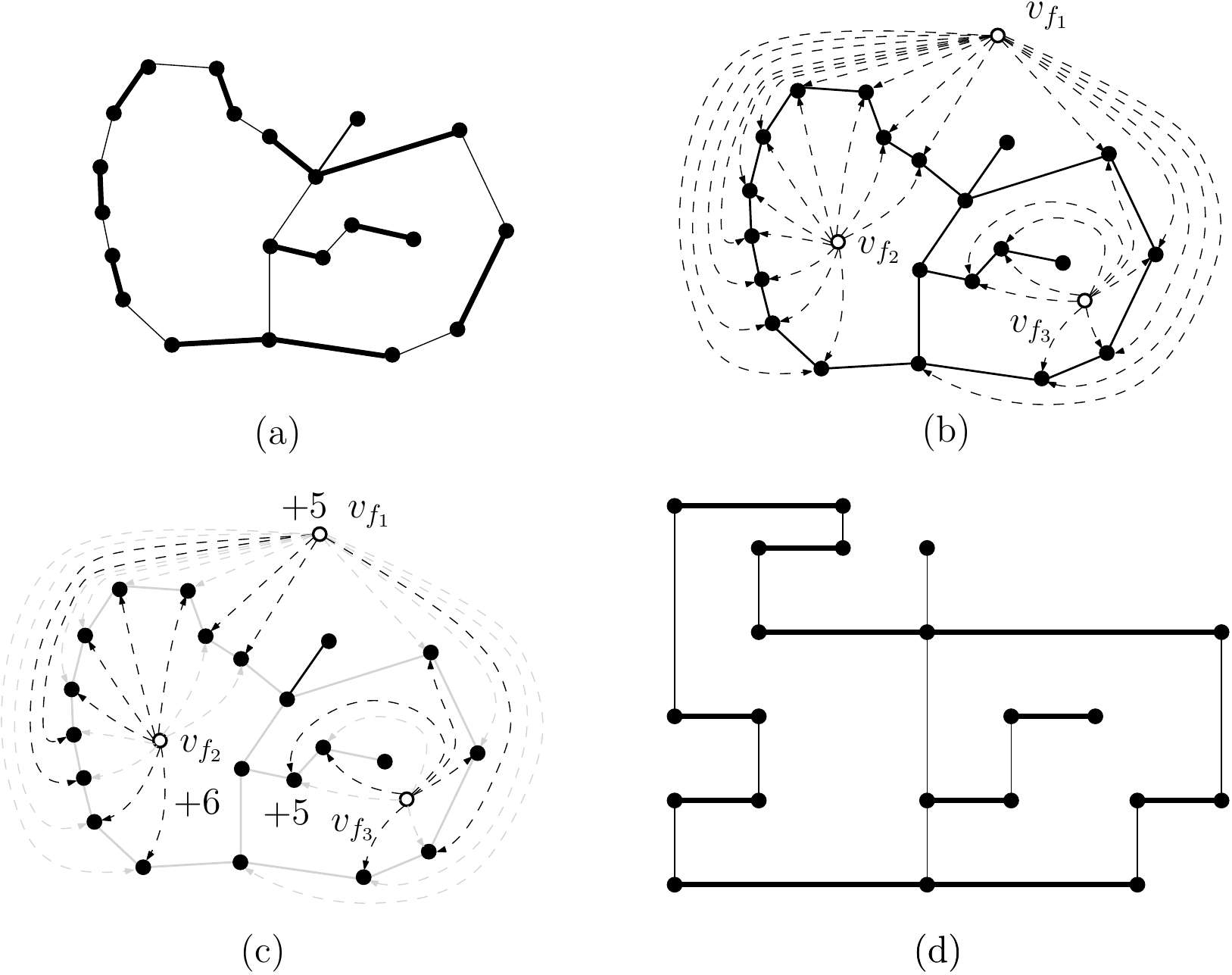}
\caption{(a) An $HV$-restricted plane graph $G$. The edges with horizontal (resp., vertical) orientations in $G$ are bold (resp., thin).  (b) The corresponding flow network  $A^\star(G)$ (induced by dashed edges).  (c) A feasible flow in $A^\star(G)$, where each black dashed edge corresponds to one unit of flow. (d) A corresponding orthogonal drawing of $G$.}
\label{fig:flow}
\end{figure}

We claim that a flow satisfying all the constraints (demand/excess/capacity)
exists if and only if $G$ admits a good orthogonal drawing preserving the input embedding.
If a flow $y \in \{0,1\}^{E_{A^\star(G)}}$ exists, then we get a solution vector for the
linear system by defining $x_c = 3 - 2y_c$ for all $e_c \in E_{A^\star(G)}$.
Together with the variables defined by conditions (a) -- (c) we obtain a global
admissible angle assignment which by definition satisfies:
\begin{enumerate}[1.]
\item The sum of angles around each vertex $v$ in $G$ is $2\pi$.
\item For every edge $(u,v)$ in $G$, the angle assignment at the
  corners of $u$ and $v$ is consistent with respect to the two faces
  that are incident to $(u,v)$.
\item The total assigned angle of every face $f$ is the angle sum
  required for polygons with that many corners. All angles are
  multiples of $\pi/2$, i.e., the induced representation is
  orthogonal.
\end{enumerate}

These conditions on an angle assignment are sufficient to construct a
plane orthogonal representation that respects the input
embedding~\cite{Tamassia87}. In fact the orthogonal drawing
can be computed in linear time. Figure~\ref{fig:flow}(d) shows an
orthogonal representation corresponding to the flow of Figure~\ref{fig:flow}(c).

For the converse, if $G$ admits a good orthogonal drawing $\Gamma$
respecting the input embedding, then the angles at the degree two
vertices readily imply a flow in the network satisfying the
constraints. We thus obtain the following theorem.
\begin{theorem}
Given an $HV$-restricted plane graph $G$ with $n$ vertices, one can check in $T(n)$ time whether 
 $G$ admits a good orthogonal drawing preserving the input embedding, and construct such a drawing 
  if it exists. Here, $T(n)$ is the time to find maximum flows in multiple-source multiple-sink directed planar graphs. 
\end{theorem}

Since the maximum flow problem for a multiple-source and multiple-sink directed planar graph can be solved in 
$O(n\log^3n)$-time~\cite{BorradaileKMNW11}, one can check whether a given
$HV$-restricted plane graph that admits a good orthogonal drawing
preserving the input embedding in $O(n\log^3 n)$ time.
 Note that we precisely know the excess and demand of each node 
 in the flow network, and hence we are actually finding a feasible flow.
 There are faster algorithms in such cases; e.g., Klein et al.~\cite{KleinMW09}
 gave an algorithm to find a feasible integral flow in $O(n\log^2n)$-time. Later,
 Mozes and Wulff-Nilsen~\cite{MozesW10} improved the running time to $O(n\log^2n/\log\log n)$.

\section{Drawing Biconnected Outerplanar Graphs}
\label{sec:planar2}
In this section we give a polynomial-time algorithm to determine whether an arbitrary biconnected $HV$-restricted outerplanar graph admits a good orthogonal drawing, and construct such a drawing if it exists. A graph is \emph{outerplanar} if it admits a planar drawing with all its vertices on the outer face. Note that the good orthogonal drawing we produce is not necessarily an outerplanar embedding.

We first show that given an $HV$-restricted planar graph $G$, one can construct a corresponding graph $G'$  with maximum degree three such that $G$ admits a good orthogonal drawing if and only if $G'$ admits such a drawing.

\begin{figure}[t]
\centering
\includegraphics[width=\textwidth]{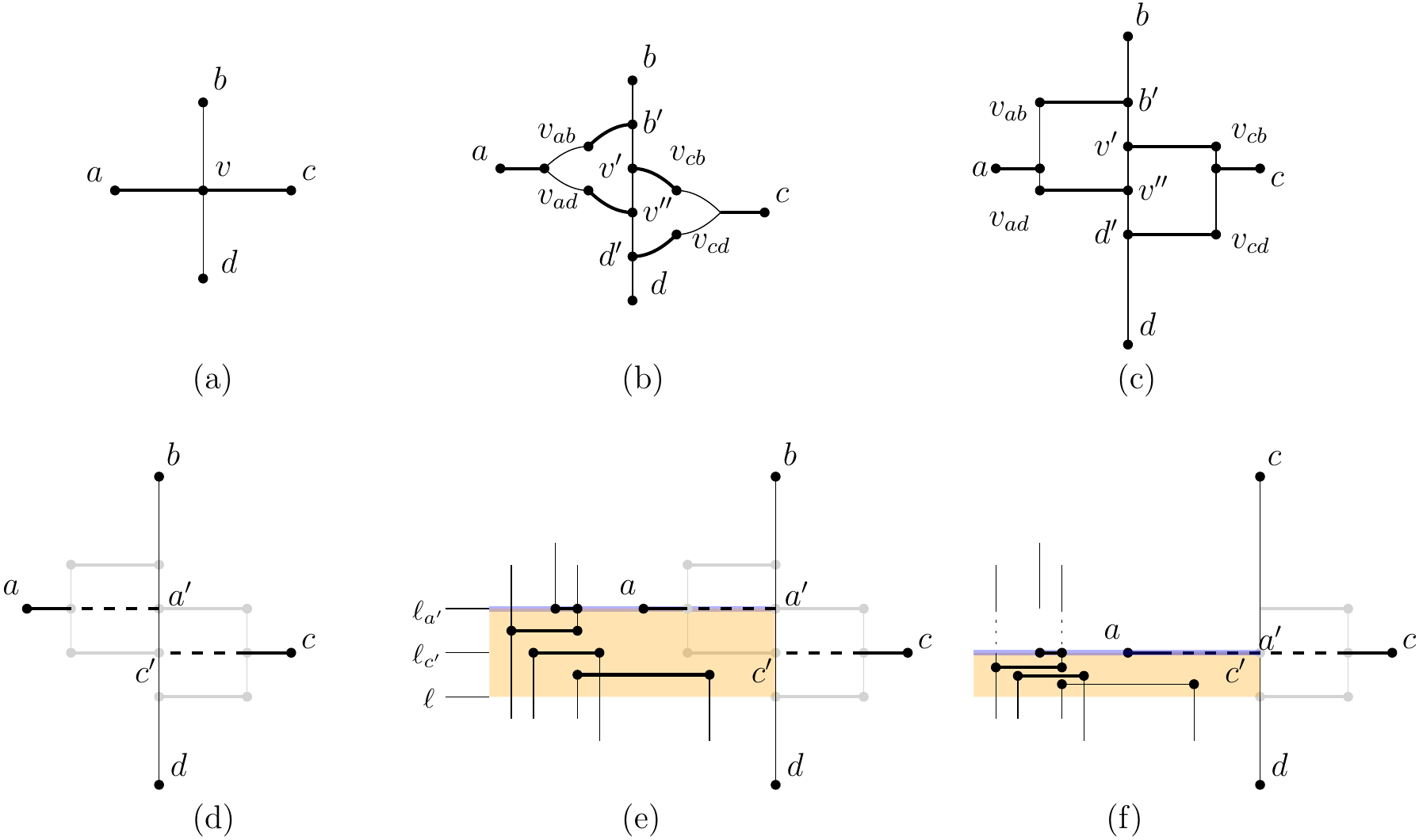}
\caption{(a) A degree four vertex $v$ with two $H$-edges and two $V$-edges. (b) A gadet to replace $v$. (c) A drawing of the gadget. (d)--(f) Transformation of a drawing of the gadget back into  a degree four vertex $v$.}
\label{fig:transformation}
\end{figure}

\begin{theorem}
\label{thm:transformation}
Let $G$ be an $HV$-restricted planar graph. For each vertex $v$ incident to two $V$-edges and two $H$-edges, replace $v$ with the vertex gadget, as illustrated in Figures~\ref{fig:transformation}(a)-(b). The resulting graph $G'$ admits a good orthogonal drawing if and only if $G$ admits a good orthogonal drawing. 
\end{theorem}
\begin{proof}
First consider that $G$ admits a good orthogonal drawing. Then for each vertex $v$ incident to two $V$-edges and two $H$-edges, we must have its $V$-edges drawn on a vertical line, e.g., see Figure~\ref{fig:transformation}(a). Therefore, it is straightforward to replace $v$ with a good drawing of the vertex gadget, as illustrated in Figure~\ref{fig:transformation}(c). Hence, $G'$ must also have a good orthogonal drawing.

Consider now that $G'$ has a good orthogonal drawing $D$. We now show how to replace every vertex gadget with its corresponding vertex, which  eventually yields a good orthogonal drawing of $G$. Let $D_v$ be a drawing of a vertex gadget in $D$. Observe that the gadget has two cycles (one attached to $a$ and the other to $c$) that force $a$ and $c$ to be drawn on opposite side of the path $P=(b,b',v',v'',d',d)$.  We project the horizontal edges incident to $a$ and $c$ onto the drawing of the path $P$, e.g., see  Figure~\ref{fig:transformation}(d). Let $a'$ and $c'$ be the corresponding points. If $a'$ and $c'$ coincide, then we take that point as the location for vertex $v$. Otherwise, without loss of generality, assume that $a'$ has a higher $y$-coordinate than $c'$. Let $\ell_{a'}$ and $\ell_{c'}$ be the  horizontal lines through $a'$ and $c'$, respectively. Let $\ell$ be a horizontal line passing through a point below $c'$ such that $\ell$ intersects segment $c'd'$ and does not contain any vertex in $D$.  We  scale down the drawing inside the region bounded by $\ell_{a'},\ell$, and $P$ vertically such that the region bounded by $\ell_{a'},\ell_{c'}$ and $P$
becomes empty. We now move the vertices and segments that lie on $\ell_{a'}$  and to the left of $a'$, on the line $\ell_{c'}$. Finally, we extend the corresponding vertical segments. Since $a'$ and $c'$ now coincide, we can place the vertex $v$ at that location.
\end{proof}

Note that by Theorem~\ref{thm:transformation}, we can restrict our attention to the maximum degree three graphs. We first introduce some notation, and then describe the characterization for maximum degree three graphs (Theorem~\ref{thm:outer}).

By a \emph{segment} of $G$, we denote a maximal path in $G$ such that all the edges on that path have the same orientation. Let $G$ be a biconnected $HV$-restricted embedded outerplanar graph  with $\Delta=3$, where $\Delta $ is the maximum degree of $G$. Let $e$ be an edge of $G$. Then by $\lambda_e$ we denote the orientation of $e$ in $G$.   Let $F$ be an inner face of $G$. Note that $G$ is an embedded graph. Thus any  edge of $G$ is an \emph{inner edge} if it does not lie on the boundary of the outer face of $G$, and all the remaining edges of $G$ are the \emph{outer edges}. An inner edge $e$ of $G$ on the boundary of $F$ is called \emph{critical} if the two edges preceding and following $e$ on the boundary of $F$ have the same orientation that is different from $\lambda_e$. 
 An edge $e$ is \emph{$h$-critical} (resp., \emph{$v$-critical})
 if it is a critical edge and $\lambda_e=H$ (resp., $\lambda_e=V$).     
 For an  inner face $F$ in $G$, let $E_v(F)$ and $E_h(F)$ be the number of distinct edges of $F$ with
 vertical and horizontal orientations, respectively.  By $C_v(F)$ and $C_h(F)$ we
 denote the number of $v$-critical and $h$-critical edges of $F$. 

Let $pqrs$ be a rectangle, and let $a$ and $b$ be two points in the proper interior of $qr$ and $rs$, respectively,
 as shown in Figures~\ref{fig:init}(a) and (b). Construct a rectangle $sbcd$, where $c$ and $d$ lie outside of the rectangle $pqrs$. Then the region consisting of the rectangles $pqrs$ and $sbcd$
 is called a \emph{flag}. A flag includes all the edges on its boundary except the edge $aq$. The rectangles $pqrs$ and $sbcd$ are called the \emph{banner} and  \emph{post}, respectively. The edges $ar$ and $br$ are called the \emph{borders} of  the flag.

\subsection{Necessary and Sufficient Conditions}     
Throughout this section, $G$ denotes an arbitrary biconnected  $HV$-restricted embedded  outerplanar graph with $\Delta=3$; see Figure~\ref{fig:init}(c) for an example. We now prove the following theorem, which  is the main result of this section. 
\begin{theorem}
\label{thm:outer}
Let $G$ be a biconnected $HV$-restricted embedded outerplanar graph with maximum degree three. Then, $G$ admits a good planar orthogonal drawing if and only if the following three conditions hold.
\begin{enumerate}[($C_1$)]
\item For every inner face $f$, the sequence of orientations of the edges in  clockwise order contains $HVHV$ as a subsequence.
\item For every inner face $f$, if $C_v(f)=E_v(f)$, then  $C_v(f)$ is even. Similarly, if $C_h(f) = E_h(f)$, then $C_h(f)$ is even.
\item Every vertex of $G$ has at most two edges of the same orientation.
\end{enumerate}
\end{theorem} 

It is straightforward to see that steps $C_1$--$C_3$ can be tested in linear time. 

\subsubsection{Necessity} We first show that Conditions ($C_1$)--($C_3$) are necessary for $G$ to admit a good planar orthogonal drawing. 
 We use the following two lemmas.
\begin{lemma}\label{lem:bc}
Let $\Gamma$ be a good orthogonal drawing of $G$, and let $(b,c)$ be an inner edge of 
 some face $f=(a,b,c,d,\ldots,a)$. Figure~\ref{fig:init}(d) illustrates
 an example. Since $(b,c)$ is an inner edge, there is another face
 $f'=(b,x,\ldots,y,c,b)$ that does not contain any edge of $f$ except $(b,c)$.
 Let $H^+$ and $H^-$ be the two half-planes determined by the straight line through $(b,c)$. If $(b,c)$ is a critical edge in $f$, then either both $(a,b)$ and $(c,d)$ lie in $H^+$, or both lie in $H^-$.
\end{lemma}
\begin{proof}
Without loss of generality assume that $\lambda_{bc} = H$. Since $(b,c)$ is a critical edge, $\lambda_{bc}\not= \lambda_{ab}$ and $\lambda_{ab}=\lambda_{cd}$. If $(a,b)$ and $(c,d)$ lie in $H^+$ and $H^-$, respectively, then one of $x$ and $y$ must lie interior to $f$ and the other must lie exterior to $f$. Therefore, the path $b,x,...,y,c$ must create an edge crossing with $f$, which contradicts that $\Gamma$ is a good orthogonal drawing.
\end{proof}

Let $x(v)$ and $y(v)$ denote the $x$- and $y$-coordinates of a vertex $v$. We now use Lemma~\ref{lem:bc} to prove the following.
\begin{lemma}
\label{lem:bc2}
Let $\Gamma$ be a good orthogonal drawing of $G$. Let $f$ be an inner face in $\Gamma$,  and let $(a,b)$ and $(c,d)$ be two edges on $f$ (without loss of generality assume that $(a,b)$ is above $(c,d)$), where  $\lambda_{ab} = \lambda_{cd} = H$, $x(a)>x(b)$ and $x(d)>x(c)$.  Let $P=(a,b,...,c,d)$ be a path on the boundary of $f$ in counter-clockwise order, e.g., see the path $P_l$ in Figure~\ref{fig:init}(e). If all the vertically oriented edges of $P$ are critical, then  the number of such critical edges on $P$ must be odd. This property holds symmetrically for the path $(b,a,...,d,c)$.
\end{lemma}
\begin{proof}
Consider a traversal of the edges of $P$ starting at $a$. Let $e$ be a $v$-critical edge on $P$, and let $e'$ and $e''$ be the edges preceding and following $e$, respectively. By Lemma~\ref{lem:bc},  $e'$ and $e''$ must lie on the same side of $e$ in $\Gamma$. Therefore, if we traverse $e'$ from left to right, then we have to traverse $e''$ from right to left, and  vice versa.  In other words, every $v$-critical edge reverses the direction of traversal. Since we traverse $(a,b)$ and $(c,d)$ from opposite directions and all the vertically oriented edges of $P$ are critical, we need  an odd number of $v$-critical edges on $P$ to complete the traversal.
\end{proof}

\begin{figure}[t]
\centering
\includegraphics[width=\textwidth]{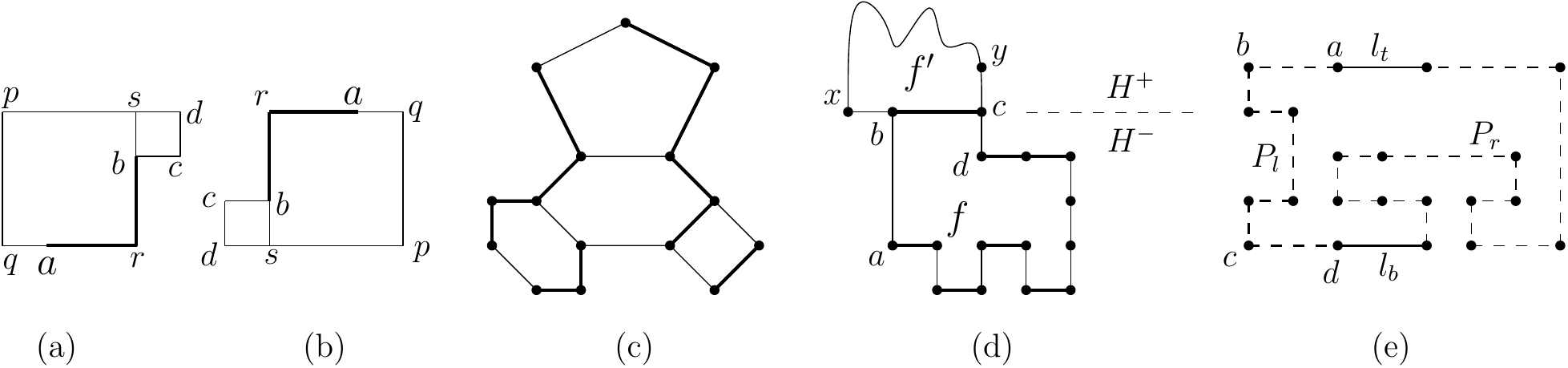}
\caption{(a)--(b) Two flags, where the borders are shown in bold. (c) An outerplanar graph $G$ with $\Delta=3$. (d) Illustration for Lemma~\ref{lem:bc}. (e) Illustration for $P_l$ and $P_r$, where $P_l$ contains three $v$-critical edges and $P_r$ contains five $v$-critical edges.}
\label{fig:init}
\end{figure}

We are now ready to prove the necessity  part of Theorem~\ref{thm:outer}. ($C_1$) holds for every orthogonal polygon, and thus for the faces of an orthogonal drawing. If ($C_2$) does not hold, then without loss of generality assume that for some $f$, $C_v(f)= E_v(f)$ and $C_v(f)$ is odd. Let $\Gamma_f$ be a drawing of $f$ such that $l_t$ and $l_b$ are topmost and bottommost horizontal edges in $\Gamma_f$. Then we can find two disjoint paths $P_l$ and $P_r$ by traversing $f$ counter-clockwise and clockwise from $l_t$ to $l_b$, respectively, as shown in Figure~\ref{fig:init}(e).  Since $C_v(f)$ is odd, either $P_l$ or $P_r$ must contain an even number of $v$-critical edges, which contradicts Lemma~\ref{lem:bc2}. If ($C_3$) does not hold at some vertex $v$, then the drawing of its incident edges would contain edge overlapping. 

\subsubsection{Sufficiency} 
\label{suff}
 To prove the sufficiency we assume that $G$ satisfies $(C_1)$--$(C_3)$, and then construct a good orthogonal drawing of $G$. The idea is to first draw an arbitrary inner face $f$ of $G$, and then the other faces of $G$ by a depth first search on the faces of $G$ starting at $f$.

Let $f=(v_1,v_2,\ldots,v_{r},\ldots,v_{s},\ldots,v_{t}(=v_1))$  be the vertices of $f$ in clockwise order. Let  $P=(v_r,\ldots,v_s,\ldots,v_t)$ be a maximal path on $f$ such that all the edges on path $P_v=(v_r,\ldots,v_s)$ (resp., $P_h=(v_s,\ldots,v_t)$) have vertical (resp., horizontal) orientation. The maximality of $P$ ensures that $\lambda_{v_1v_2}=V$ and $\lambda_{v_{r-1}v_r}=H$. An example of such a path $P$ in the face of Figure~\ref{fig:suf}(a) is $a(=v_r),b,c(=v_s),d,e(=v_t)$. Observe that $\lambda_{v_1v_2} = \lambda_{eg} =V$ and $\lambda_{v_{r-1}v_r} = \lambda_{ia} =H$. 

The following technical lemma allows us to draw the graph $G$; we prove Lemma~\ref{cor:p} in Section~\ref{sec:detail}.
\begin{lemma}
\label{cor:p}
Given an inner face $f$ of $G$ that satisfies conditions $(C_1)$--$(C_3)$, and a drawing of two consecutive  segments $P_h$ and $P_v$ of $f$. One can find a good orthogonal drawing $\Gamma_f$ of $f$ that satisfies the following properties.
\begin{enumerate}[-]
\item Lemma~\ref{lem:bc} holds for every critical edge $e$ in $\Gamma_f$, i.e., the two edges preceding and following $e$ lie in the same side of $e$.
\item $\Gamma_f$ is contained in a flag $F$ with borders $P_h$ and $P_v$.
\item If $P_h$ is a critical edge, then the post of $F$ (if exists) is incident to $P_v$. Similarly,  if $P_v$ is a critical edge, then the post of $\Gamma_f$ (if exists) is incident to $P_h$. (Note that since $\Delta = 3$, both $P_h$ and $P_v$ cannot be critical). 
\end{enumerate}
\end{lemma}

\begin{figure}[t]
\centering
\includegraphics[width=\textwidth]{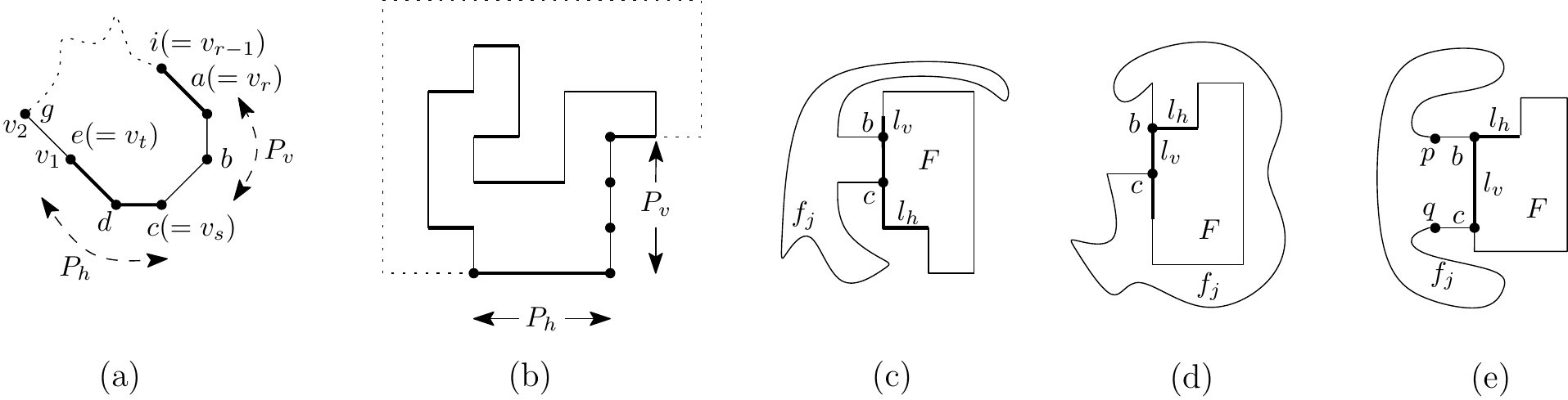}
\caption{(a) An inner face of $G$. (b) Illustration for $\Gamma_f$. (c)--(e) Illustration for the construction of  $\Gamma_{k+1}$.}
\label{fig:suf}
\end{figure}

We are now ready to describe the drawing of  $G$; hence, completing the proof of Theorem~\ref{thm:outer}. We first construct the drawing $\Gamma_f$ for some inner face $f$ of $G$. We then draw the other inner faces of $G$ by a depth first search on the faces of $G$, such that after adding a new inner face, the resulting drawing remains 
\begin{enumerate}[$(P_1)$]
\item a good orthogonal drawing, and
\item each critical edge respects Lemma~\ref{lem:bc}.
 \end{enumerate}
Let $\Gamma_k$ be  a drawing of the set of inner faces $f_1(=f),f_2,\ldots, f_{k}$ that we have already constructed. Let $f_{k+1}$ be an inner face of $G$ that has not been drawn yet, but has an edge $(b,c)$ in common with some face $f_j$, where $1\le j\le k$. Without loss of generality assume that $\lambda_{bc}=V$ in $\Gamma_k$. Furthermore, since $G$ is outerplanar, $f_{k+1}$ cannot have any edge other than $(b,c)$ in common with $f_j$. Let $l_v$ be a segment of $f_{k+1}$ that contains $(b,c)$, and let $l_h$ be another segment of $f_{k+1}$ incident to $l_v$.  We now construct $\Gamma_{k+1}$ considering the following cases.
\begin{description}
\item [Case 1: none of $b$ and $c$ is an end vertex of $l_v$.] In this case none of the end vertices of the path formed by $l_v$ and $l_h$ belongs to $\Gamma_k$.    
Since $G$ satisfies Condition ($C_3$), the edges of $f_j$ that are incident to $b$ and $c$ must be horizontal, i.e., $(b,c)$ must be a $v$-critical edge of $f_j$. Since $\Gamma_k$ is a good orthogonal drawing, there is enough space to create a flag $F$ with borders $l_v$ and $l_h$ such that the banner and  post of $F$ do not create any edge crossing. Figure~\ref{fig:suf}(c) illustrates such an example. By Lemma~\ref{cor:p}, we can draw $f_{k+1}$ inside $F$ maintaining Properties $(P_1)$ and $(P_2)$. Thus the resulting drawing $\Gamma_{k+1}$ satisfies ($P_1$)--($P_2$).

\item [Case 2: exactly one of $b$ and $c$ is an end vertex  of $l_v$.]  If $b$ (resp., $c$) is an end vertex of $l_v$, then we choose $l_h$ such that it contains $b$ (resp., $c$). Therefore, none of the two end vertices of the path formed by $l_v$ and $l_h$ belongs to $\Gamma_k$. Figure~\ref{fig:suf}(d) illustrates such an example. Similar to Case 1, we now draw $\Gamma_{k+1}$ satisfying $(P_1)$--$(P_2)$. 

\item [Case 3: both $b$ and $c$ are end vertices of $l_v$.] Observe that in this case $l_v=(b,c)$. Let $a,b,c,d$ be a path of $f_{k+1}$. Since $l_v=(b,c)$ is a maximal set of edges with vertical orientation, we have $\lambda_{ab}= \lambda_{bc}=H$. Thus  $l_v=(b,c)$ is a $v$-critical edge of $f_{k+1}$. We now create a flag $F$ with borders $l_v$ and $l_h$ such that the post of the flag is incident to $l_h$. Note that since $l_v$ is critical, by Lemma~\ref{cor:p}, we do not require a flag with its post incident to $l_v$. We now can draw $f_{k+1}$ inside $F$ maintaining $(P_2)$ and $(P_3)$.  Figure~\ref{fig:suf}(e) illustrates such an example. It may initially appear from the figure that drawing of $f_{k+1}$ inside $F$ may overlap the boundary of $f_j$, i.e., consider the Figure~\ref{fig:suf}(e) with $\lambda_{cq}=V$.  However, by definition of a flag, $F$ does not contain the part of its boundary that overlaps $f_j$, and hence drawing $f_{k+1}$ would not create any  edge overlapping.
\end{description}

\section{Proof of Lemma~\ref{cor:p}}
\label{sec:detail}
In this section, we prove Lemma~\ref{cor:p}. We first show the following auxiliary lemma, which proves the first condition of Lemma~\ref{cor:p}. 
\begin{lemma}
\label{lem:star}
If $f$ satisfies Conditions $(C_1)$--$(C_3)$, then we can find a good orthogonal  drawing of $f$ such that Lemma~\ref{lem:bc} holds for every critical edge of $f$. 
\end{lemma}

\begin{proof}
Our proof is constructive. In the following we first construct a drawing $\Gamma_{f}$ of $f$, and then prove that $\Gamma_{f}$ is the required good orthogonal drawing. 

\subsection{Construction of $\Gamma_{f}$}
Since $f$ satisfies $(C_1)$, $P$ must contain at least three vertices (recall the definition of $P$ from Section~\ref{suff}). We first draw the path $P$ maintaining edge orientations, as shown in Figure~\ref{fig:drawing}(a). Let the drawing be $\Gamma_P$. We next draw $P'=(v_1,v_2,\ldots,v_r)$. 
 We draw  $P'$ starting at vertex $v_t(=v_1)$ of $\Gamma_P$, and then complete the drawing of $P'$ such that the position of $v_r$ coincides with its position in $\Gamma_P$. The details of the construction is presented in Steps 1--3 given below.

\paragraph{Step 1: satisfying $v$-critical edges.} Recall that $\lambda_{v_1v_2}=V$ and $\lambda_{v_{r-1}v_r}=H$. We now draw $P'$ starting at $v_1$ in $\Gamma_P$ with monotonically increasing $y$-coordinate and maintaining the edge orientations. Furthermore, we ensure that for each vertical segment $l$, the two horizontal segments incident to $l$ lie on the same side of $l$, e.g., see Figure~\ref{fig:drawing}(b). Hence, Lemma~\ref{lem:bc} holds for every $v$-critical edge $e$ in $P'$. It is straightforward to draw $P'$ without crossing any edge of $\Gamma_P$. However, as shown in Figure~\ref{fig:drawing}(b), the position of $v_r$ in the drawing of $P'$ may not coincide with its position in $\Gamma_P$. Let the resulting drawing of $P'$ be $\Gamma_{P'}$. Observe that $\Gamma_{P'}$ respects  the following property.
\begin{enumerate}[(A)]
 \item Lemma~\ref{lem:bc} holds for every $v$-critical edge $e$ in $\Gamma_{P'}$, i.e.,
 the edges preceding and following $e$ lie on the same side of $e$.
\end{enumerate}

\begin{figure}[t]
\centering
\includegraphics[width=0.9\textwidth]{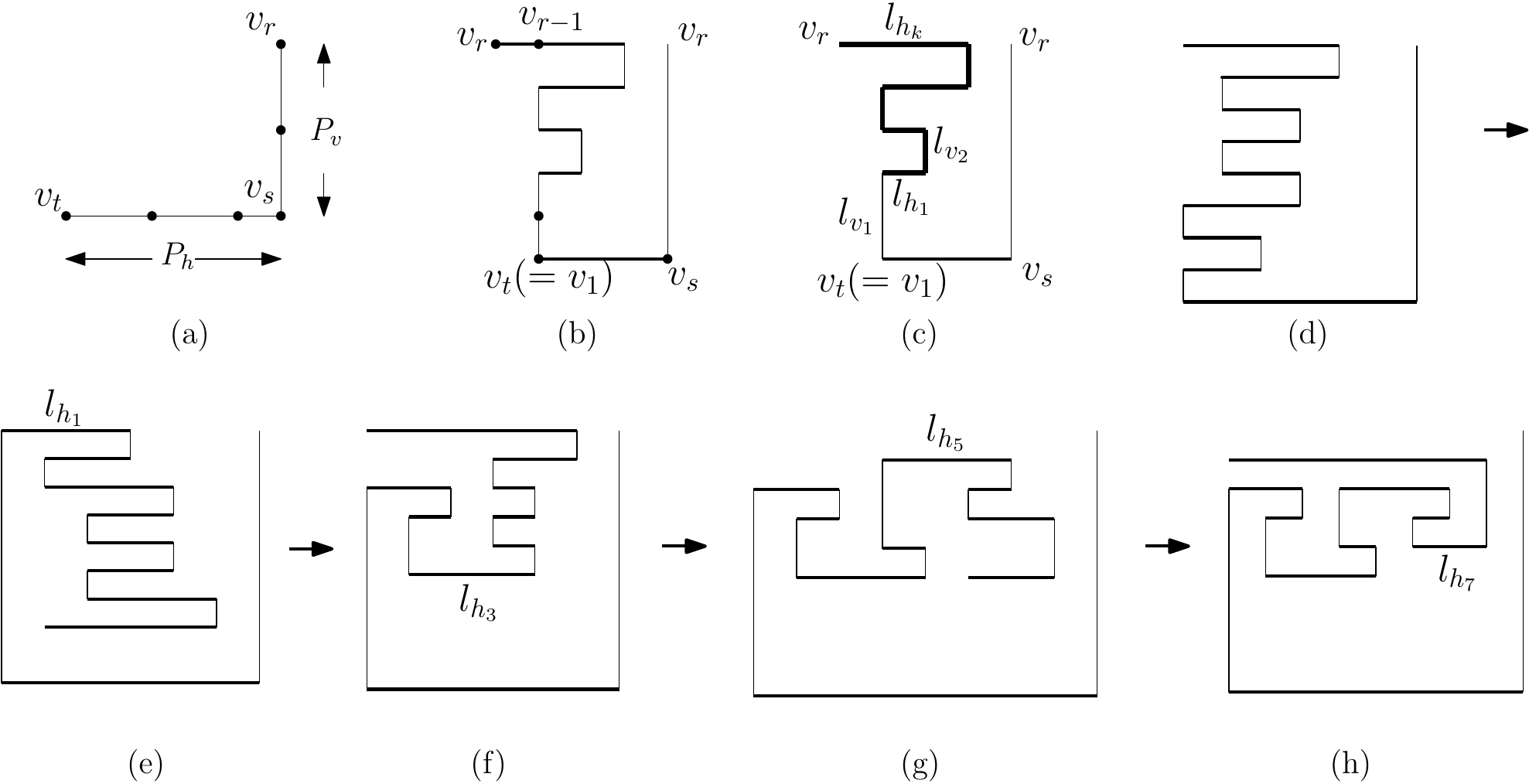}
\caption{(a) $\Gamma_P$. (b) Initial drawing of $P'$. (c) Illustration for Step $2$.}
\label{fig:drawing}
\end{figure}

\paragraph{Step 2: satisfying $h$-critical edges.} For every segment $l$ in $P'$, let $\Psi_{l}$ be the drawing of the subpath that starts at segment $l$ and ends at $v_r$. Let $l_{v_1},l_{h_1},l_{v_2},l_{h_2},...,$ $l_{v_k},l_{h_k}$ be the segments of $P'$ in clockwise order, where $l_{v_i}$ and $l_{h_i}$, $1\le i\le k$, denote the vertical and horizontal segments, respectively. For example, see Figure~\ref{fig:drawing}(c), where $\Psi_{l_{h_1}}$ is shown in bold. We now modify $\Gamma_{P'}$ in three phases.

\paragraph{Phase 1: fabricating.} For each $i$ from $1$ to $k$, where $i$ is odd, we flip  $\Psi_{l_{h_i}}$ with respect to $l_{h_i}$. For example, see Figures~\ref{fig:drawing}(d)-(h). It is straightforward to compute such a flip avoiding edge crossings by adjusting the edge lengths as necessary. Consequently, if $l_{h_i}$ is an $h$-critical edge, then Lemma~\ref{lem:bc} holds for $l_{h_i}$; i.e., the edges preceding and following $l_{h_i}$ lie on the same side of $l_{h_i}$. Observe that such vertical flips do not destroy Property (A).

We now consider $l_{v_i}$, where $i$ is even. For each $i$ from $2$ to $k$, where $i$ is even, if $l_{h_i}$ is an $h$-critical edge, then we first flip $l_{h_i}$ with respect to $l_{v_i}$, and then flip $\Psi_{l_{h_i}}$ with respect to $l_{h_i}$. Consequently, if $l_{h_i}$ is an $h$-critical edge, then Lemma~\ref{lem:bc} holds for $l_{h_i}$. Figure~\ref{fig:phase}(a)  illustrates an example.

Since we flip $l_{h_i}$ horizontally, it may initially appear that we are destroying Property (A) when $l_{v_i}$ is a critical edge. However, observe that we flip $l_{h_i}$ only if it is an $h$-critical edge. Since the maximum degree of $G$ is three,  $l_{v_i}$ cannot be a critical edge. Therefore, Property (A) still holds.

While modifying $\Gamma_{P'}$, it is straightforward to ensure that the downward infinite ray starting at $v_r$ do not cross any segment except possibly $P_h$. Furthermore, the infinite ray starting at $l_{h_k}$ towards right would intersect only $P_v$. Let the resulting drawing of $P'$ be $\Gamma'_{P'}$. By the construction of $\Gamma'_{P'}$, we can observe the following properties.
\begin{enumerate}
\item[(B)] The downward infinite ray in $\Gamma'_{P'}$ starting at $v_r$ do not cross any segment, except possibly $P_h$. The infinite ray starting at $l_{h_k}$ towards right intersects only $P_v$.
\item[(C)] Lemma~\ref{lem:bc} holds for every critical edge $e$ in $\Gamma'_{P'}$, except possibly $(v_{r-1},v_r)$.
\end{enumerate}

\begin{figure}[t]
\centering
\includegraphics[width=0.9\textwidth]{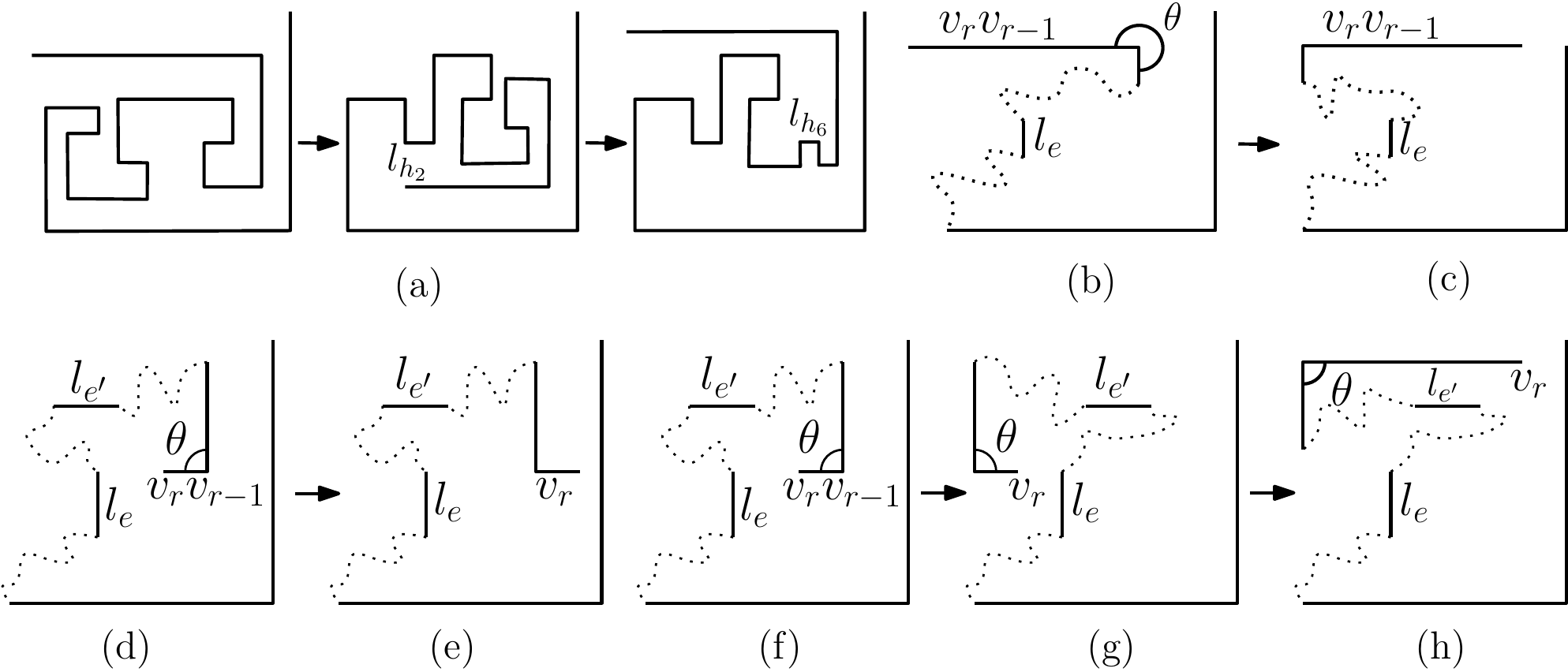}
\caption{(a) Illustration for Phase 1. (b)--(h) Illustration for Phase 2.}
\label{fig:phase}
\end{figure}

\paragraph{Phase 2: vertical adjustment.} Let $\theta$ be the angle interior to $f$ formed by segments $l_{v_k}$ and $l_{h_k}$. 
If $x(v_{r-1}) > x(v_r)$ and $\theta>\pi/2$  (e.g., Figure~\ref{fig:phase}(b)), then we find a non-critical vertical edge $e$ in $P'$.
 If there exists such an edge $e$, then let $l_e$ be the segment that contains $e$, and  we flip 
 $\Psi_{l_e}$ with respect to $l_e$  such that  $x(v_{r-1}) < x(v_r)$ holds, e.g., see Figures~\ref{fig:phase}(b)--(c).
 We can adjust the length of the segments such that the property $(B)$ still holds. 
 If $e$ does not exist, then we do not modify $\Gamma'_{P'}$, i.e., $x(v_{r-1}) > x(v_r)$ still holds.\smallskip

If $x(v_{r-1}) > x(v_r)$ and $\theta=\pi/2$  (e.g., Figure~\ref{fig:phase}(d)), then we find a non-critical vertical edge $e$ and a non-critical horizontal edge $e'$ in $P'\setminus l_{h_k}$. If $e'$ does not exist, then all horizontal edges are critical, and hence $l_{v_k}$ cannot be critical. Here we flip $l_{h_k}$ with respect to $l_{v_k}$, as shown in  Figures~\ref{fig:phase}(d)--(e).  Otherwise, if both $e$ and $e'$ exist, then let $l_e$ and $l_{e'}$ be the segments that contain $e$ and $e'$, respectively. We first flip $\Psi_{l_e}$ with respect to $l_e$  such that  $x(v_{r-1}) < x(v_r)$ holds. We then flip $\Psi_{l_e'}$ with respect to $l_e'$  such that Property (B) holds.
 For example, see Figures~\ref{fig:phase}(f)--(h).
 If $e$ does not exist, then we do not modify $\Gamma'_{P'}$, i.e., $x(v_{r-1}) > x(v_r)$ still holds.\smallskip 
 
Note that if $x(v_{r-1}) > x(v_r)$ holds even after the above analysis, then  all the vertical edges of $P'$ are critical; i.e., $\Gamma'_{P'}$ now satisfies the following additional property.

\begin{enumerate}[(D)]
\item If $x(v_{r-1}) > x(v_r)$ holds in $\Gamma'_{P'}$, then all of its vertical edges  are critical.
\end{enumerate}

\paragraph{Phase 3: horizontal adjustment.} Consider the segment $l_{v_k}$, and let $v_p$ and $v_q$ be its end vertices, where $v_q$ is closer to $v_r$ than $v_p$ in $P'$. If $y(v_p) > y(v_q)$, e.g., Figures~\ref{fig:phase3}(a) and (c), then we find a non-critical horizontal edge $e$ in $P'\setminus l_{h_k}$. If there exists such an edge $e$, then let $l_e$ be the segment that contains $e$, and  we flip $\Psi_{l_e}$ with respect to $l_e$  such that  $y(v_p) < y(v_q)$ holds. We can adjust the length of the segments such that the property $(B)$ still holds. Figures~\ref{fig:phase3}(b) and (d) illustrate such modifications. Otherwise, $e$ does not exist, and $y(v_p) > y(v_q)$ holds. $\Gamma'_{P'}$ now satisfies the following additional property.
\begin{enumerate}[(E)]
\item If $y(v_p) > y(v_q)$ holds in $\Gamma'_{P'}$, then all of its horizontal edges, except possibly $l_{h_k}$,  are critical.  If there exists a non-critical horizontal edge $e$ in $P'\setminus l_{h_k}$, then $y(v_p) < y(v_q)$.
\end{enumerate}

\begin{figure}[t]
\centering
\includegraphics[width=0.9\textwidth]{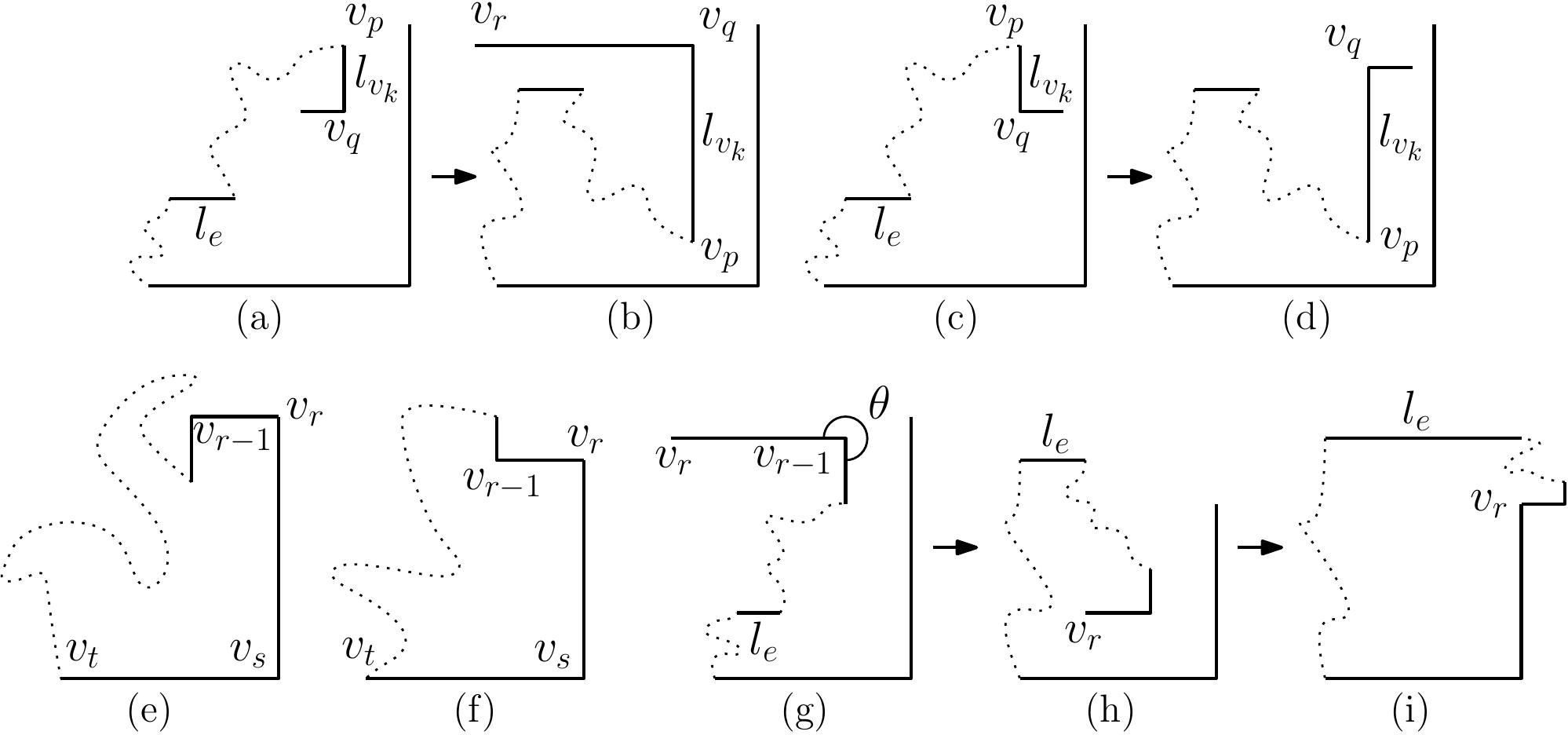}
\caption{ (a)--(d) Illustration for Phase 3. (e)--(i) Illustration for Step 3.}
\label{fig:phase3}
\end{figure} 
  
\paragraph{Step 3: completing the cycle.} We now modify $\Gamma'_{P'}$ such that the position of $v_r$ in $\Gamma'_{P'}$ and the position of $v_r$  in $\Gamma_{P}$ coincide. We first  check whether $x(v_{r-1}) < x(v_r)$ or not.

If $x(v_{r-1}) < x(v_r)$, then we scale up the drawing of $P'$ vertically and then extend the edge $(v_{r-1},v_r)$ such that the position of $v_r$ in $\Gamma'_{P'}$ and $\Gamma_p$ coincide, as shown in Figures~\ref{fig:phase3}(e)--(f). By Property (B), we can perform these modifications avoiding edge crossings. The resulting drawing is $\Gamma_{f}$, which still satisfies Property (C). Note that $\Gamma_{f}$ lies inside a flag with border $P_h$ and $P_v$, where the post of the flag is degenerate.
 
 Otherwise, $x(v_{r-1}) > x(v_r)$. Then by Property (D), all the vertical edges of $\Gamma'_{P'}$ are critical.
 Since $G$ is of maximum degree three, all the horizontal edges are non-critical. 
 By  Property (E), 
 $y(v_p)<y(v_q)$.
 This scenario is illustrated in Figure~\ref{fig:phase3}(g). Let $e$ be a horizontal non-critical edge. We then flip the drawing  $\Psi_{l_e}$ with respect to $l_e$, and then adjust the drawing such that the positions of $v_r$ in $\Gamma'_{P'}$ and $\Gamma_{P}$ coincide. An example is shown in Figures~\ref{fig:phase3}(g)--(i). The resulting drawing is $\Gamma_{f}$, which still satisfies Property (C) and the condition that $x(v_{r-1}) > x(v_r)$.  Note that $\Gamma_{f}$ lies inside a flag with border $P_h$ and $P_v$, where the post of the flag is incident to $P_v$.

\subsection{$\Gamma_{f}$ satisfies Lemma~\ref{lem:star}}
We now show that $\Gamma_{f}$ 
is a good orthogonal drawing and Lemma~\ref{lem:bc} holds for each of its critical edges. 
 By construction,  $\Gamma_{f}$ 
 is planar. To see that $\Gamma_{f}$ 
 respects edge orientations, observe that the drawing  $\Gamma_{P'}$ obtained from Step $1$ respects edge orientations. Later, Steps $2$--$3$ only modify the drawing by flipping some segments vertically and horizontally, which does not destroy the edge orientations. We now prove that Lemma~\ref{lem:bc} holds for every critical edge. 
\begin{description}
\item [Case 1: $v$-critical edge on $P'$.] Since $\Gamma'_{P'}$ satisfies Property (C), Lemma~\ref{lem:bc} holds for every $v$-critical edge on $P'$.
\item [Case 2: $v$-critical edge on $P_v$.] In this case $P_v$ consists of only a single edge, i.e., $(v_r,v_s)$. If the angle interior to $f$ at $v_r$ is equal to $\pi/2$, e.g., see Figures~\ref{fig:phase3}(e)--(f), then Lemma~\ref{lem:bc} holds for $(v_r,v_s)$. Otherwise, the  angle is greater than $\pi/2$, i.e., $x(v_{r-1})>x(v_r)$, as shown in Figure~\ref{fig:phase3}(i). According to Property (D), all the vertical edges on $P'$ must be critical. Hence, $C_v(f)= E_v(f)$. We show that this case cannot appear by proving that $C_v(f)$ must be odd, i.e., $f$ violates Condition $(C_2)$. Hence Lemma~\ref{lem:bc} must hold.

Let $e_t\in P'$  be a  topmost and let and $e_b\in P_h$ bottommost horizontal edge in $\Gamma_f$. We then  can find two disjoint paths $P_l$ and $P_r$ by traversing $f$ counter-clockwise and clockwise from $e_t$ to $e_b$, respectively. Since $C_v(f)= E_v(f)$, by Lemma~\ref{lem:bc2}, $P_l$ and $P_r\setminus (v_r,v_s)$ each must have odd number of $v$-critical edges. Since $(v_r,v_s)$ itself is a critical edge, the total number of $v$-critical edges; i.e., $C_v(f)$ must be odd and hence $f$ violates Condition $(C_2)$.

\item [Case 3: $h$-critical edge on $P'$.] By Property (C), Lemma~\ref{lem:bc} holds for every $h$-critical edge on $P'$, except possibly $(v_{r-1},v_r)$. Consider now the case when $(v_{r-1},v_r)$ is $h$-critical. Let $v_p$ and $v_q$ be the end vertices of $l_{v_{k-1}}$, where $v_q$ is closer to $v_r$ in $P'$. If $y(v_{p})< y(v_{q})$: e.g., see Figure~\ref{fig:phase3}(e), then Lemma~\ref{lem:bc} holds. 

Now consider the case when $y(v_{p})> y(v_{q})$ (i.e., Figures~\ref{fig:phase3}(f) and (i)). If $x(v_{r-1})>x(v_r)$, then by Property (D), all the vertical edges would be critical and thus we will not have any $h$-critical edge in $P'$. Hence we only need to consider the case when   $x(v_{r-1})<x(v_r)$ (i.e., Figure~\ref{fig:phase3}(f)). In the following we show that in such a scenario $f$ must violate Condition $(C_2)$. Hence Lemma~\ref{lem:bc} must hold. 
  
Since $y(v_{p})> y(v_{q})$ and $x(v_{r-1})<x(v_r)$, by Property (E), all the horizontal segments of $\Gamma'_{P'}$, except possibly $l_{h_k}$, are $h$-critical. Since $l_{h_k}$ is a critical edge,  all the horizontal  segments on $P'$ are $h$-critical. Using an argument similar to Case 2 (i.e., by choosing a leftmost and a rightmost vertical edge and using Lemma~\ref{lem:bc2}), we can observe that the number of $h$-critical edges excluding  $l_{h_k}$ is even. Since $l_{h_k}$ is $h$-critical, $C_h(f)$ must be odd; i.e., $f$ violates Condition $(C_2)$.
\end{description}
\end{proof}

We are now ready to prove Lemma~\ref{cor:p}.
\begin{proof}[Proof of Lemma~\ref{cor:p}]
Let $\Gamma_f$ be a drawing produced by Steps 1--3. Then by Lemma~\ref{lem:star}, $\Gamma_f$ is a good orthogonal drawing and Lemma~\ref{lem:bc} holds for every critical edge in $\Gamma_f$. We now show that $\Gamma_f$ satisfies the remaining conditions of Lemma~\ref{cor:p}, i.e., 
\begin{enumerate}[-]
\item $\Gamma_f$ is contained in a flag $F$ with borders $P_h$ and $P_v$.
\item If $P_h$ is a critical edge, then the post of $F$ (if exists) is incident to $P_v$. Similarly,  if $P_v$ is a critical edge, then post of $F$ (if exists) is incident to $P_h$. 
\end{enumerate}  

The first property is implied by the construction of $\Gamma_f$ as follows. Steps $1$--$2$ ensure that the drawing of $P'$ lies inside the banner of some flag with borders $P_h$ and $P_v$, whereas Step 3 creates the post of the flag, e.g., see Figure~\ref{fig:phase3}(i). 

The second property holds because Lemma~\ref{lem:bc} holds for every critical edge of $\Gamma_f$. Specifically, if $e$ is a critical edge, where either $e=P_v$ or $e=P_h$ and the post is incident to $e$, then the two horizontal  segments incident to $e$ will lie on different sides of $e$; i.e., Lemma~\ref{lem:bc} would be violated.
\end{proof}

\section{Conclusion}
\label{sec:con}
In this paper, we gave a polynomial-time algorithm to decide good orthogonal drawability of $HV$-restricted plane graphs, and moreover, fully characterized $HV$-restricted biconnected outerplanar graphs that admit good orthogonal drawings. If we relax the biconnectivity constraint, then our characterization no longer holds. For example, the $HV$-restricted outerplanar graph $G$ of Figure~\ref{fig:hardex}(b) satisfies Conditions $(C_1)$-$(C_3)$ of Theorem~\ref{thm:outer}, but does not admit a good orthogonal drawing.

\begin{figure}[t]
\centering
\includegraphics[width=0.85\textwidth]{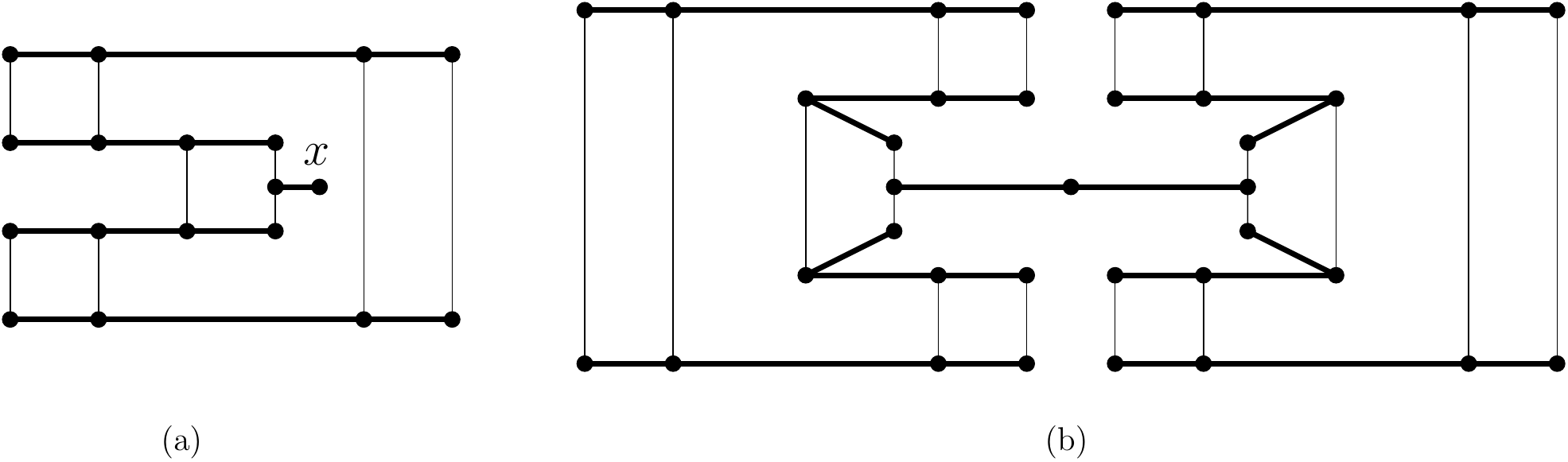}
\caption{Illustration for the graphs (a) $H$ and (b) $G$.}
\label{fig:hardex}
\end{figure}

Observe that  $G$ is constructed from two copies of the graph $H$ of Figure~\ref{fig:hardex}(a), where the vertices with label $x$ are identified. Since in any good orthogonal drawing of $H$ the vertex $x$  lies in some inner face, any orthogonal drawing of $G$ preserving edge orientations must contain edge crossing. Hence, a natural open question is to extend the characterization for arbitrary outerplanar graphs.

\bibliographystyle{plain}
\bibliography{ref}

\begin{thebibliography}{10}

\bibitem{DBLP:journals/comgeo/AlamKM17}
Md.~Jawaherul Alam, S.~G. Kobourov, and D.~Mondal.
\newblock Orthogonal layout with optimal face complexity.
\newblock {\em Comput. Geom.}, 63:40--52, 2017.

\bibitem{battista}
G.~Di Battista, G.~Liotta, and F.~Vargiu.
\newblock Spirality and optimal orthogonal drawings.
\newblock {\em SIAM J. Comput.}, 27:1764--1811, 1998.

\bibitem{BorradaileKMNW11}
G.~Borradaile, P.~N. Klein, S.~Mozes, Y.~Nussbaum, and C.~Wulff{-}Nilsen.
\newblock Multiple-source multiple-sink maximum flow in directed planar graphs
  in near-linear time.
\newblock {\em {SIAM} J. Comput.}, 46(4):1280--1303, 2017.

\bibitem{CornelsenK12}
S.~Cornelsen and A.~Karrenbauer.
\newblock Accelerated bend minimization.
\newblock {\em J. Graph Algorithms Appl.}, 16(3):635--650, 2012.

\bibitem{BattistaKLLW12}
G.~{Di Battista}, E.~Kim, G.~Liotta, A.~Lubiw, and S.~Whitesides.
\newblock The shape of orthogonal cycles in three dimensions.
\newblock {\em Discrete {\&} Computational Geometry}, 47(3):461--491, 2012.

\bibitem{GiacomoLP02}
E.~{Di Giacomo}, G.~Liotta, and M.~Patrignani.
\newblock Orthogonal {3D} shapes of theta graphs.
\newblock In {\em Proceedings of the International Symposium on Graph Drawing
  (GD)}, volume 2528 of {\em LNCS}, pages 142--149. Springer, 2002.

\bibitem{DidimoLP19}
W.~Didimo, G.~Liotta, and M.~Patrignani.
\newblock {HV}-planarity: Algorithms and complexity.
\newblock {\em J. Comput. Syst. Sci.}, 99:72--90, 2019.

\bibitem{DurocherF0M14}
S.~Durocher, S.~Felsner, S.~Mehrabi, and D.~Mondal.
\newblock Drawing {HV}-restricted planar graphs.
\newblock In {\em proceedings of the 11th Latin American Symposium on
  Theoretical Informatics (LATIN)}, volume 8392 of {\em LNCS}, pages 156--167.
  Springer, 2014.

\bibitem{garg2001}
A.~Garg and R.~Tamassia.
\newblock On the computational complexity of upward and rectilinear planarity
  testing.
\newblock {\em SIAM J. Comput.}, 31(2):601--625, 2001.

\bibitem{hoffmann1988}
F.~Hoffmann and K.~Kriegel.
\newblock Embedding rectilinear graphs in linear time.
\newblock {\em Information Processing Letters}, 29(2):75--79, 1988.

\bibitem{kant96}
G.~Kant.
\newblock Drawing planar graphs using the canonical ordering.
\newblock {\em Algorithmica}, 16:4--32., 1996.

\bibitem{KleinMW09}
P.~N. Klein, S.~Mozes, and O.~Weimann.
\newblock Shortest paths in directed planar graphs with negative lengths: {A}
  linear-space {$O(n \log ^2 n)$}-time algorithm.
\newblock {\em {ACM} Trans. Algorithms}, 6(2):30:1--30:18, 2010.

\bibitem{manuch2010}
J.~Ma\v{n}uch, M.~Patterson, S.-H. Poon, and C.~Thachuk.
\newblock Complexity of finding non-planar rectilinear drawings of graphs.
\newblock In Ulrik Brandes and Sabine Cornelsen, editors, {\em Proceedings of
  the 18th International Symposium on Graph Drawing (GD)}, volume 6502 of {\em
  LNCS}, pages 305--316. Springer, 2010.

\bibitem{MiuraHN06}
K.~Miura, H.~Haga, and T.~Nishizeki.
\newblock Inner rectangular drawings of plane graphs.
\newblock {\em Int. J. Comput. Geometry Appl.}, 16(2--3):249--270, 2006.

\bibitem{MozesW10}
S.~Mozes and C.~Wulff{-}Nilsen.
\newblock Shortest paths in planar graphs with real lengths in {$O(n\log ^2
  n/\log\log n)$} time.
\newblock In {\em Proceedings of the 18th Annual European Symposium on
  Algorithms (ESA)}, volume 6347 of {\em LNCS}, pages 206--217. Springer, 2010.

\bibitem{nomura}
K.~Nomura, S.~Tayu, and S.~Ueno.
\newblock On the orthogonal drawing of outerplanar graphs.
\newblock {\em IEICE Trans. Fundamentals}, E88-A:1583--1588, 2005.

\bibitem{Tamassia87}
R.~Tamassia.
\newblock On embedding a graph in the grid with the minimum number of bends.
\newblock {\em SIAM J. Comput.}, 16(3):421--444, 1987.

\bibitem{vijayan1985}
G.~Vijayan and A.~Wigderson.
\newblock Rectilinear graphs and their embeddings.
\newblock {\em SIAM J. Computing}, 14(2):355--372, 1985.

\end{thebibliography}

\end{document}